\documentclass{article}

\usepackage{microtype}
\usepackage{graphicx}
\usepackage{subfigure}
\usepackage{booktabs} % for professional tables
\usepackage{times}
\usepackage{soul}
\usepackage[small]{caption}
\usepackage{graphicx}
\usepackage{amsmath}
\usepackage{algorithm}
\usepackage{algorithmic}
\usepackage{mathptmx}
\usepackage{amsmath}
\usepackage{amsthm}

\newtheorem{proposition}{Proposition}
\newtheorem{definition}{Definition}
\newtheorem{theorem}{Theorem}
\usepackage{subfigure}
\usepackage[export]{adjustbox}
\usepackage{hyperref}
\usepackage{appendix}
\usepackage{amsfonts}
\usepackage{multirow}
\usepackage{amssymb}

\usepackage{hyperref}

\usepackage[accepted]{icml2020}

\icmltitlerunning{Q-value Path Decomposition for Deep Multiagent Reinforcement Learning}

\begin{document}

\twocolumn[
\icmltitle{Q-value Path Decomposition for Deep Multiagent Reinforcement Learning}




\begin{icmlauthorlist}
\icmlauthor{Yaodong Yang}{tju}
\icmlauthor{Jianye Hao}{tju}
\icmlauthor{Guangyong Chen}{tencent}
\icmlauthor{Hongyao Tang}{tju}
\icmlauthor{Yingfeng Chen}{netease}
\icmlauthor{Yujing Hu}{netease}
\icmlauthor{Changjie Fan}{netease}
\icmlauthor{Zhongyu Wei}{fdu}
\end{icmlauthorlist}

\icmlaffiliation{tju}{Tianjin Univeristy}
\icmlaffiliation{tencent}{Tencent}
\icmlaffiliation{netease}{NetEase Fuxi AI Lab}
\icmlaffiliation{fdu}{Fudan University}
\icmlcorrespondingauthor{Jianye Hao}{jianye.hao@tju.edu.cn}

\icmlkeywords{Machine Learning, ICML}

\vskip 0.3in
]



\printAffiliationsAndNotice{}  % leave blank if no need to mention equal contribution

\begin{abstract}
Recently, deep multiagent reinforcement learning (MARL) has become a highly active research area as many real-world problems can be inherently viewed as multiagent systems. A particularly interesting and widely applicable class of problems is the partially observable cooperative multiagent setting, in which a team of agents learns to coordinate their behaviors conditioning on their private observations and commonly shared global reward signals. One natural solution is to resort to the centralized training and decentralized execution paradigm. During centralized training, one key challenge is the multiagent credit assignment: how to allocate the global rewards for individual agent policies for better coordination towards maximizing system-level's benefits. In this paper, we propose a new method called Q-value Path Decomposition (QPD) to decompose the system's global Q-values into individual agents' Q-values. Unlike previous works which restrict the representation relation of the individual Q-values and the global one, we leverage the integrated gradient attribution technique into deep MARL to directly decompose global Q-values along trajectory paths to assign credits for agents. We evaluate QPD on the challenging StarCraft II micromanagement tasks and show that QPD achieves the state-of-the-art performance in both homogeneous and heterogeneous multiagent scenarios compared with existing cooperative MARL algorithms.
%The source codes can be accessed by the anonymous link: \url{https://github.com/QPD-NeurIPS2019/QPD}.
\end{abstract}

\section{Introduction}
Cooperative multiagent reinforcement learning problem has been studied extensively in the last decade \cite{busoniu_comprehensive_2008,gupta_cooperative_2017,palmer_lenient_2018}, where a system of agents learn towards coordinated policies to optimize the accumulated global rewards. Cooperative multiagent systems (MAS) have been shown to be a useful paradigm in numerous applications, e.g., the coordination of autonomous vehicles \cite{Cao2012An} and optimizing the productivity of a factory in distributed logistics \cite{Ying2005Multi}.

One natural way of addressing cooperative MARL problem is the centralized approach, which views the overall MAS as a whole and solves it as a single-agent learning task. In such settings, existing reinforcement learning (RL) techniques can be leveraged to learn joint optimal policies based on agents’ joint observations and common rewards \cite{tan_multi-agent_1993}. However, the centralized approach usually does not scale well, since the joint action space of agents grows exponentially as the increase of the number of agents.
%Another related side-effect is called “lazy agent phenomenon” \cite{Sunehag2017Value}, which refers to that some agents may become lazy by learning inefficient policies eventually.
Furthermore, centralized approaches may not be applicable in practical settings where only distributed policies can be deployed due to physical observation and communication constraints \cite{foerster_counterfactual_2018}, i.e., each agent can only decide to behave based on its local observations.

To address these above limitations, an alternative technique is to resort to decentralized approaches, in which each agent learns its optimal policy independently based on its local observations and individual rewards. However, in cooperative multiagent environments, all agents receive the same global reward signal. Letting individual agents learn concurrently based on the global reward (aka. independent learners) has been well studied \cite{tan_multi-agent_1993} and shown to be difficult in even simple two-agent, single-state stochastic coordination problems. One main reason is that the global reward signal brings the nonstationarity that agents cannot distinguish between the stochasticity of the environment and explorative behaviors of other co-learners \cite{lowe_multi-agent_2017}, and thus may mistakenly update their policies. Therefore, the key to promoting the coordination of agents is to correctly allocate the reward signal for each agent, which is also known as the multiagent credit assignment problem \cite{chang_all_2004}.

For simple problems, it might be possible to manually design the individual reward function for each agent based on domain knowledge. However, the heuristic design requires manual efforts and is not always applicable in complex cooperative multiagent tasks. It would be desirable if there is any generalized principle or mechanism to generate individual reward functions in a universal and automatic manner. Foerster et al. \cite{foerster_counterfactual_2018} proposed a multiagent actor-critic method called counterfactual multiagent (COMA) policy gradients, which uses a counterfactual baseline that marginalizes out a single agent’s action while keeping the other agents’ actions fixed to calculate the advantage for agent policies. Sunehag et al. \cite{sunehag_value-decomposition_2018} proposed a value-decomposition network (VDN), which learns to decompose the team value function into agent-wise value functions. However, this work assumes that the joint action-value function for the system can be decomposed into the sum of agents' value functions only based on local observations. Such an assumption is not applicable for complex systems where agents have complicated relations and the decomposition is not accurate as the global information is not fully utilized. Based on VDN, QMIX relaxes the limitation of the linear relation of global Q-values ($Q_{tot}$) and the local individual Q-values ($Q^i$), which enforces a monotonicity constraint on the relationship between $Q_{tot}$ and each $Q^i$. QMIX employs a network that estimates joint action-values as a complex non-linear combination of per-agent values that condition only on local observations.

However, VDN and QMIX both restrict the relation representation between the individual Q-values and the global Q-value while the individual Q-values are only estimated from local observations. Such a way restricts the accuracy of the individual Q-values and may impede the learning of coordinated policies in complex multiagent scenarios. Recently, QTRAN \cite{son_qtran_2019} is proposed to guarantee optimal decentralization by using linear constraints between individual utilities and $Q_{tot}$, and avoids the representation limitations introduced by VDN and QMIX. But the constraints on the optimization problem are computationally intractable and practical relaxations lead to unsatisfied performance in complex tasks \cite{mahajan_maven_2019}.

% In addition to that VDN, COMA and QMIX aim to address the multiagent credit assignment problem. There have seen many related contributions on the decentralized partially observable Markov decision process (Dec-POMDPs) problems. For the large-scale MAS setting, Duc Thien Nguyen et al., \cite{nguyen_policy_2017,nguyen_credit_2018} study the Collective Dec-POMDPs where agent interactions are dependent on their collective influence on each other rather than their identities. At the same time, Yang et al., \cite{yang_mean_2018} assume that each agent is affected by its neighbors to reduce the nonstationary phenomenon and derive a mean-field approach. For the multiagent exploration problem, Mahajan et al., \cite{mahajan_maven_2019} propose MAVEN, where value-based agents condition their behaviour on the shared latent variable controlled by a hierarchical policy. Additionally, by treating Dec-POMDPs as Occupancy-MDPs, oSARSA is proposed in \cite{dibangoye_learning_2018} and a policy-gradient approach is proposed in \cite{bono_cooperative_2019}. As they focus on different MARL challenges and are applicable for different task scenarios, we do not discuss them in details.

In this paper, we propose a novel Q-value decomposition technique from the perspective of deep learning (DL). Similar to previous works, we set in a centralized learning and decentralized execution paradigm, where agents are trained centrally with shared information while executing in a decentralized manner. Our method employs integrated gradients \cite{pmlr-v70-sundararajan17a} to analyze the contribution of each agent to the global Q-value $Q_{tot}$, and regards the contribution of each agent as its individual $Q^i$, which is used as the supervision signal to train each agent's Q-value function. As we utilize trajectories of RL to implement attribution decomposition, we call this method Q-value Path Decomposition (QPD). 
%QPD can be abstracted by the following ideas. First, similar to previous works, QPD focuses on a practical setting called centralized learning and decentralized execution, where agents learn decentralized policies for execution while training centrally with shared information such as the joint local observations and actions. While executing, each agent can only obtain its local observation and cannot communicate with others. 
Besides, we design a multi-channel critic to generate $Q_{tot}$ by following the individual, group and system concepts progressively based on agents' joint observations and actions. Lastly, we merge the integrated gradients into RL to decompose $Q_{tot}$ into approximative $Q^i$ with respect to each agent's local observation and action for precise credit assignment. We evaluate QPD using the StarCraft II micromanagement tasks. Experiments show that QPD learns effective policies in both homogeneous and heterogeneous scenarios with the state-of-the-art performance.

There have seen many related contributions on the setting of the decentralized partially observable Markov decision process (Dec-POMDPs). For the large-scale MAS setting, Duc Thien Nguyen et al., \cite{nguyen_policy_2017,nguyen_credit_2018} study the Collective Dec-POMDPs where agent interactions are dependent on their collective influence on each other rather than their identities. At the same time, Yang et al., \cite{yang_mean_2018} assume that each agent is affected by its neighbors to reduce the nonstationary phenomenon and derive a mean-field approach. Above two methods are only investigated in the large-scale multiagent settings and satisfy the theoretical support under the large-scale assumption. Another notable direction is the multiagent exploration problem. Mahajan et al., \cite{mahajan_maven_2019} propose MAVEN to solve it, where value-based agents condition their behaviour on the shared latent variable controlled by a hierarchical policy. Their latent space which controls the exploration of joint behaviours mainly affects on the agent's individual utility network and is orthogonal to ours.

The remainder of this paper is organized as follows. We introduce the Dec-POMDPs and integrated gradients in Section~\ref{section:background}. Then in Section~\ref{section:framework}, we explain our QPD framework for deep MARL in details. Next, we validate our methods in the challenging StarCraft II platform in Section~\ref{section:experiment}. Finally, conclusions and future work are provided in Section~\ref{section:conclusion}. %and related works are briefly introduced in Section~\ref{section:related_work}.

%We evaluate our algorithm with the StarCraft\uppercase\expandafter{\romannumeral2} as the use case to explore the learning of intelligent collaborative strategies among multiple agents, which has recently emerged as a challenging RL benchmark task with high stochasticity, a large state-action space, and delayed rewards. In more detail, we focus on StarCraft\uppercase\expandafter{\romannumeral2} micromanagement tasks, where each player controls their own units to destroy the opponents army in the combats under different terrain conditions. To demonstrate the effectiveness of QPD, we follow the challenging game settings of \cite{rashid_qmix:_2018}, which massively reduces each agent’s field-of-view and removes access to these macro-actions to make the game more realistic and intricate. Our experiments show that QPD helps multi-agent systems learn complicated cooperation policies in both heterogeneous and asymmetric scenarios.

\section{Background}
\label{section:background}
\subsection{Dec-POMDPs}
Fully cooperative multiagent tasks can be modeled as Dec-POMDPs \cite{oliehoek_concise_2016}. Formally, a Dec-POMDP $G$ is given by a tuple
\begin{equation}
G = <S, A, P, r, Z, O, n, \gamma>
\label{MDP}
\end{equation}
where $s \in S$ describes the true state of the environment. Dec-POMDPs consider partially observable scenarios in which an observation function $Z(s, i): S \times N \rightarrow p(O)$, which defines the probability distribution of the observations $o^i \in O^i$ for each agent $i \in N \equiv \{1,...,n\}$ draws individually. At each time step, each agent $i$ selects its action $a_{i} \in A_i$ based on its local observation $o_{i}$ according to its stochastic policy $\pi_{i}:O_i \times A_i \rightarrow [0,1]$. The joint action $\vec{a} \in \vec{A}$ produces the next state according to the state transition function $P:S \times A_1 \times ... \times A_n \rightarrow S$. All agents share the same reward function $r(s,\vec{a}):S \times \vec{A} \to {R}$. All agents coordinate together to maximize the total expected return $J = {E}_{a_{1}\sim\pi_{1}, ..., a_{n}\sim\pi_{n}, s\sim  P} \sum_{t=0}^T \gamma^t r_{t}(s,\vec{a})$ where $\gamma$ is a discount factor and T is the time horizon. Our problem setting follows the paradigm of centralized training and decentralized execution \cite{foerster_counterfactual_2018}. That is, each agent executes its policy in a distributed manner, since agents may only observe the partial environmental information due to physical limitations (e.g., scope or interferer) and high communication cost in practice. However, each agent's policy can be trained in a centralized manner (using a simulator with additional global information) to improve the learning efficiency. The global discounted return is $R_{t}=\sum_{l=0}^{T-t} \gamma^{l} r_{t+l}$. The agents' joint policy induces a value function, i.e., an approximation of expectation over $R_{t}$, $V^{\vec{\pi}}(s_{t}) = {E}_{\vec{a}_{t} \sim \vec{\pi}, s_{t+1} \sim P}[R_{t}|s_{t}]$, and a global action-value $Q^{\vec{\pi}}(s_{t}, \vec{a}_{t}) = {E}_{\vec{a}_{t+1} \sim \vec{\pi}, s_{t+1} \sim P}[R_{t}|s_{t},\vec{a}_{t}]$  remarked as $Q_{tot}$.

\subsection{Integrated Gradients}
A lot of works intend to understand the input-output behavior of the deep network and attribute the prediction of a deep network to its input features \cite{ancona2018towards}. The goal of attribution methods is to determine how much influence does each component of input features have in the network output value \cite{braso2018attribution}.
\begin{definition}
Formally, suppose we have a function $F:\mathbb{R}^d \rightarrow \mathbb{R}$ that represents a deep network, and an input $x=(x_1,...,x_j,...,x_d) \in \mathbb{R}^d$. $\mathbb{R}$ is the set of real number. $F$ is the function with a $d$-dimension vector input. An attribution of the prediction at input $\vec{x}$ relative to a baseline input $\vec{b}$ is a vector $A_{F}(\vec{x},\vec{b})=(c_1,...c_j,...,c_d) \in \mathbb{R}^d$, where $c_j$ is the contribution value of $x_j$ to the difference between prediction $F(\vec{x})$ and the baseline prediction $F(\vec{b})$.
\end{definition}

The attribution methods are widely studied \cite{Baehrens_2010,binder_layer-wise_2016,montavon_methods_2018}. As one of them, integrated gradients takes use of path integral to aggregate the gradients along the inputs that fall on the lines between the baseline and the input \cite{pmlr-v70-sundararajan17a}, which is inspired by economic cost-sharing literature \cite{tarashev_risk_2016} with theoretical supports \cite{hazewinkel_encyclopaedia_1990}. The integrated gradients explains how much one feature affects the deep network output while changing from $F(\vec{b})$ to $F(\vec{x})$ along a straight line between $\vec{x}$ and $\vec{b}$. Although integrated gradients uses the straightline, there are many paths that monotonically interpolate between the two points, and each such path will yield a different attribution method depicting the feature changing process. The path integral focuses on the changing process of each variable to perform attribution and has shown impressive performance.
%and can be naturally adopted into RL (a path corresponds to a trajectory) by applying integrated gradients on discrete time steps.

Formally, let $\tau(\alpha):[0,1]\rightarrow R^d$ be a smooth path function specifying a path in $R^d$ from the baseline $\vec{b}$ to the input $\vec{x}$, i.e., $\tau(0)=\vec{b}$ and $\tau(1)=\vec{x}$. Given a path function $\tau$, path integrated gradients are obtained by integrating gradients along the path $\tau(\alpha)$ for $\alpha \in [0,1]$. Mathematically, path integrated gradients along the $j$th dimension for input $\vec{x}$ is defined as follows.
\begin{equation}
\label{eq:pathig}
    c_{j} = PathIG^{\tau}_{j}(\vec{x}) ::= \int^{1}_{\alpha=0}\frac{\partial F(\tau(\alpha))}{\partial \tau_j(\alpha)}\frac{\partial \tau_j(\alpha)}{\partial \alpha} \mathrm{d}\alpha,
\end{equation}
where $\frac{\partial F(\tau(\alpha))}{\partial \tau_j(\alpha)}$ is the gradient of $F$ along the $j$th dimension.

Attribution methods based on path integrated gradients are collectively known as path methods. Sundararajan et al,. first introduce path integrated gradients to perform attribution for the deep network. Due to the absence of the real feature varying path, they specify the straightline as the path for integration. Using the straightline path $\tau(\alpha)=\vec{b}+\alpha(\vec{x} - \vec{b})$ for $\alpha \in [0,1]$, the integrated gradients \cite{pmlr-v70-sundararajan17a} to calculate the contribution value $c_{j}$ along the $j$th dimension for input $\vec{x}$ is defined as follows.
\begin{equation}
\label{eq:ig}
    c_{j} = IG^{\tau}_{j}(\vec{x})::= (\vec{x}_j - \vec{b}_j) \int^{1}_{\alpha=0}\frac{\partial F(\tau(\alpha))}{\partial \tau_j(\alpha)} \mathrm{d}\alpha.
\end{equation}
In the computer vision and natural language processing domains, when applying integrated gradients, the zero embedding vector is usually used as the baseline $\vec{b}$. Besides, as mentioned above, the straightline is the choice for the path. It seems there are no better path choices for the image models or natural language models as the feature varying process is unknown.
%The reason why integrated gradients uses the straightline as the path is that it seems no better path choices for image models or natural language models as the feature varying process is unknown. 
The zero-vector baseline and corresponding straightline are not suitable for many real problems as they do not really reflect how features change. For example, in an episode of RL, transition of state and action features happens between every two adjacent steps from time $t$ to $T$. Such a feature varying process cannot be depicted by the straightline from the starting state to the all-zero vector.

\section{QPD for MARL}
\label{section:framework}
Here we describe our QPD MARL framework and Figure~\ref{figure:framework} shows the overall learning framework. First, in Section~\ref{sec:critic}, we design a centralized critic which consists of modular channels to extract hidden states for different groups of agents to learn the global Q-value $Q_{tot}$ from agents' joint observations and actions. Then we leverage integrated gradients techniques on the multiagent multi-channel critic to decompose $Q_{tot}$ into individual Q-values $Q^i$ approximately for each agent in Section~\ref{sec:decomposition}. Such a decomposition process addresses the multiagent credit assignment via the covariation analysis of each agent's observations and actions along the trajectory path. The decomposed individual value which approximates $Q^i$ is used as the supervision signal to train each agent's recurrent Q-value network. Finally, we give the algorithm details and training losses in Section~\ref{sec:alg}.

\begin{figure}[ht]
\centering
{\includegraphics[height=4.0in,angle=0]{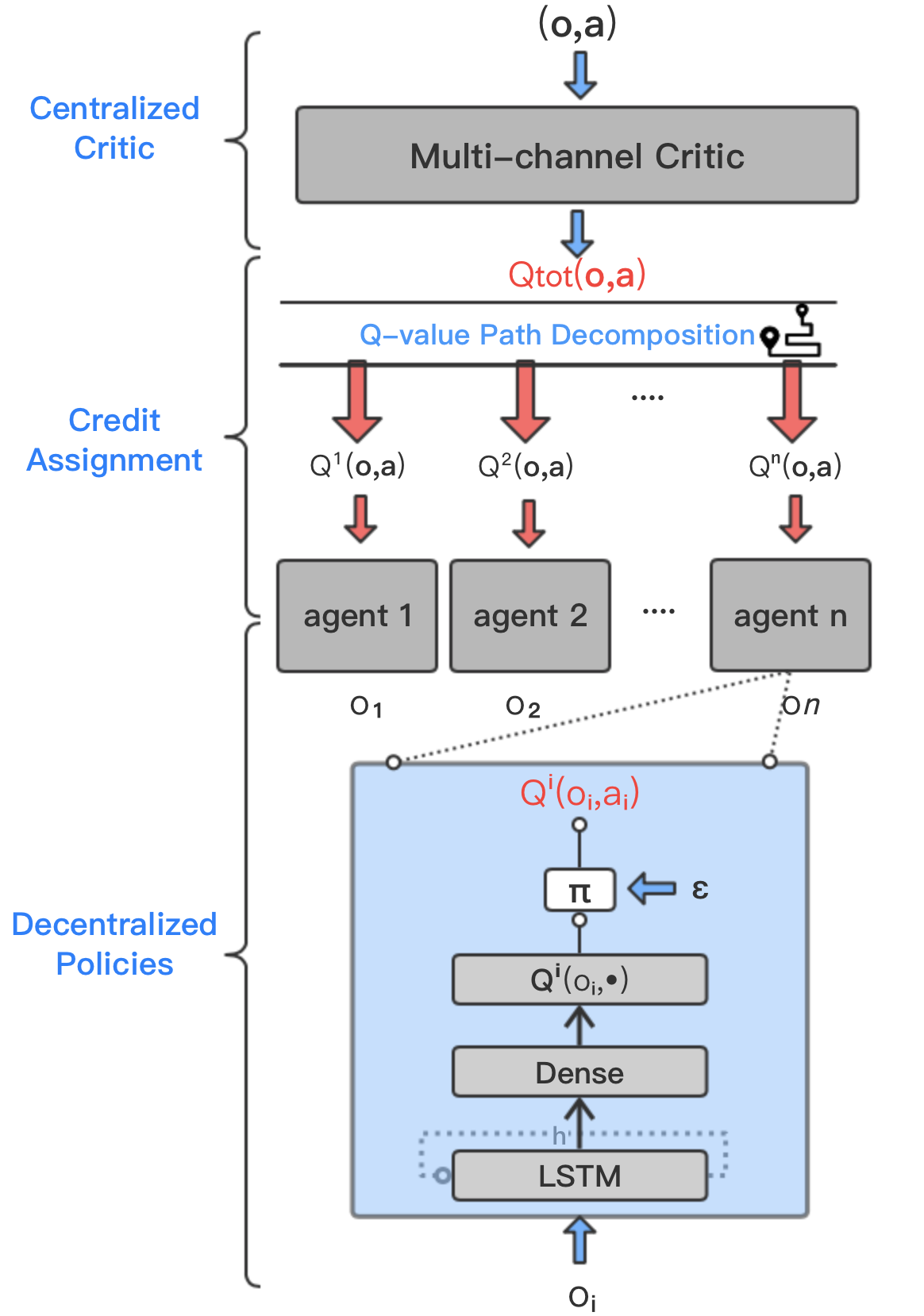}}
\caption{The overall QPD Framework. The top block is the centralized critic with a multi-channel modular design. The middle block is applying the Q-value path decomposition technique to achieve credit assignments on the agent level. The $Q_{tot}$ is decomposed into the supervision signals for $Q^{i}$. The bottom block shows the network architecture of the agent policies, which are implemented by the recurrent deep Q-network.}
\label{figure:framework}
\end{figure}

\subsection{Value Decomposition Through Integrated Gradients}
\label{sec:decomposition}
In this section, we apply integrated gradients to assigning credits for each agent on the multi-channel critic by performing attribution on each own states and actions with respect to the output $Q_{tot}$. As DRL employs deep neural networks to approximate the global $Q_{tot}$, we could utilize the attribution tools in DL combined with concepts in RL to extract the contribution of specified sets of features from different agents to the predicted Q-value. To this end, in this paper, we propose a new multiagent credit assignment approach which utilizes integrated gradients on the state-action trajectory.

As mentioned above, in DL, it is usually unknown how features change from input to baseline. Thus, the straightline becomes the path choice for using integrated gradients in DL. However, in RL, a natural path luckily exists, which is the trajectory of state-action transitions in each episode and records how the state-action features change. As the trajectory path depicts the real feature varying process, we could achieve an accurate attribution. Using integrated gradients on the trajectory, we perform the global Q-value decomposition by attributing the global Q-value prediction to its input features. Applying integrated gradients into RL was first studied in RUDDER \cite{arjona-medina_rudder:_2018} to address the sparse delayed reward problem in single-agent RL and has shown excellent performance. However, one weakness of their approach is that they regard the zero vector as the baseline for all states, which ignores real state-action transitions. Another limitation is that they do not use the trajectory but the straightline between current states and the zero vector as the path when applying integrated gradients, thus making the decomposition inaccurate. Different from RUDDER, we utilize the basic trajectory concept in RL to avoid above issues and then use integrated gradients to naturally conduct multiagent credit assignment.

%Considering the partial observation problem and the communication limitation in reality situations, we adopt the centralized training, decentralize execution paradigm \cite{foerster_counterfactual_2018,rashid_qmix:_2018} as the algorithm training scheme to achieve better algorithm performance than the independent learners. In such a scheme, each agent has its own policy based on its individual partial observation, and the updating of each policy depends on a common critic. The common critic receives all agents' observations and actions, and can give accurate updating signals for each agent when only global reward can be obtained. Like \cite{foerster_counterfactual_2018,rashid_qmix:_2018}, we allow the accessibility of global state while training to help learning, which can be obtained while training in the simulators.

Now we introduce how to use the path integrated gradients on trajectories to decompose $Q_{tot}$ into approximative individual $Q^i$. The key to path integrated gradients is to find the correct changing path of each agent's state-action features. As we analyzed previously, such a path could be depicted by the state-action transition trajectory in an RL environment, which captures the state-action feature transformation process from the start state to the termination state. Besides, with the trajectory as the path, we can naturally use the termination state $s_T$ as baseline where $Q(s_T, \varnothing)=0$. $\varnothing$ means no action is further taken at the termination state.
%The episode trajectory naturally address above problems by using state-action transitions to identify the feature changing path and using the terminated state $S_T$ as the baseline where $Q(S_T, \cdot)=0$.
After specifying both the integration path and baseline, we employ integrated gradients on the trajectory path to decompose the critic's prediction $Q_{tot}$ to each agent's local observations and actions to implement the credit assignment. Formally, using joint observations $\vec{o}$ to represent the global state $s$, we have Equation~\ref{eq:decomposeig} and the proof is provided in Theorem~\ref{prop:add}.
\begin{equation}
\small
\begin{split}
    %& \;\, Q_{tot}(s_t,\vec{a_t}) = Q_{tot}(s_t,\vec{a_t}) - Q_{tot}(s_T,\cdot) = Q_{tot}(\vec{o_t},\vec{a_t})\\
    & Q_{tot}(\vec{o_t},\vec{a_t})
     % = Q^1(\vec{o_t},\vec{a_t}) + ... + Q^n(\vec{o_t},\vec{a_t}) \\ %= \sum_i^n Q^i(\vec{o_t},\vec{a_t}) \\
     = \sum_{x_j\in \mathbb{X}_{1}}PathIG^{\tau_{t}^{T}}_{j}(\vec{o_t},\vec{a_t}) + ... + \sum_{x_j\in \mathbb{X}_{n}}PathIG^{\tau_{t}^{T}}_{j}(\vec{o_t},\vec{a_t}). \\
    %& = \sum_i^n Q^i(\vec{o_t},\vec{a_t}) (\mbox{see proof in Proposition~\ref{prop:add}}), \; where \\
    %&(\mbox{see proof in Theorem~\ref{prop:add}}.)
\end{split}
\label{eq:decomposeig}
\end{equation}
where $\tau_{t}^{T}$ is the trajectory path from time $t$ to $T$, and every two adjacent joint observations and actions are connected by straightlines. $\mathbb{X}_{i}$ is the set of agent $i$'s observation features and action dimensions. By decomposing the global Q-value following the real trajectory path, we get each agent's individual contribution to $Q_{tot}$ based on its own observation and action. Because the attribution reveals how much each agent's own observation and action contributes to $Q_{tot}$ by following the real trajectory path, we regard the attribution value $\sum_{x_j\in \mathbb{X}_{i}}PathIG^{\tau}_{j}(\vec{o_t},\vec{a_t})$ of agent $i$'s observation-action features as its approximative individual Q-value $Q^i(\vec{o_t},\vec{a_t})$.
% \begin{equation}
% \begin{split}
%     & Q^i(\vec{o_t},\vec{a_t}) \approx \sum_{x_j\in \mathbb{X}_{i}}PathIG^{\tau_{t}^{T}}_{j}(\vec{o_t},\vec{a_t}) \\ & = \sum_{x_j\in \mathbb{X}_{i}}IG^{\tau_t^{t+1}}_{j}(\vec{o},\vec{a})
%      + \sum_{x_j\in \mathbb{X}_{i}}IG^{\tau_{t+1}^{t+2}}_{j}(\vec{o},\vec{a}) + 
%      ... + \sum_{x_j\in \mathbb{X}_{i}}IG^{\tau_{T-1}^{T}}_{j}(\vec{o},\vec{a}).
% \end{split}
% \end{equation}
\begin{equation}
    Q^i(\vec{o_t},\vec{a_t}) \approx \sum_{x_j\in \mathbb{X}_{i}}PathIG^{\tau_{t}^{T}}_{j}(\vec{o_t},\vec{a_t}).
\end{equation}
Then the next question is how to compute $\sum_{x_j\in \mathbb{X}_{i}}PathIG^{\tau_{t}^{T}}_{j}(\vec{o_t},\vec{a_t})$. As paths between every two adjacent joint observations and actions are straightlines in the path $\tau_{t}^{T}$, we can directly apply integrated gradients on the line between every adjacent joint observations and actions from $(\vec{o}_{t+1},\vec{a}_{t+1})$ to $(\vec{o}_{t},\vec{a}_{t})$ as shown in Equation~\ref{eq:pathig2ig}.
\begin{equation}
\small
\label{eq:pathig2ig}
\begin{split}
    & \sum_{x_j\in \mathbb{X}_{i}}PathIG^{\tau_{t}^{T}}_{j}(\vec{o_t},\vec{a_t}) =
    \\ & \sum_{x_j\in \mathbb{X}_{i}}IG^{\tau_t^{t+1}}_{j}(\vec{o},\vec{a})
     + \sum_{x_j\in \mathbb{X}_{i}}IG^{\tau_{t+1}^{t+2}}_{j}(\vec{o},\vec{a}) + 
     ... + \sum_{x_j\in \mathbb{X}_{i}}IG^{\tau_{T-1}^{T}}_{j}(\vec{o},\vec{a}).
\end{split}
\end{equation}
Using integrated gradients to decompose $Q_{tot}$ makes the most of the available global information while previous works such as VDN and QMIX compute individual $Q^i$ from agents' local observations and actions and limit the accuracy of $Q^i$. %Using the joint observations and actions could make the decomposition more accurate.
Next, in Theorem~\ref{prop:add}, we prove that decomposing global $Q_{tot}$ through the trajectory satisfies the additive property across agents, which realizes an intact decomposition.
%Specifically, the attributions among all agents add up to the global Q-value, as formally stated in Equation~\ref{eq:valueinvariance}. 
Before proof, we introduce one important property of integrated gradients that the attributions add up to the difference between function $F$'s outputs at the input $\vec{x}$ and baseline $\vec{b}$, which will be used in proving Theorem~\ref{prop:add}.
%which is formalized by the proposition below.
\begin{proposition}
If $F:R^d \leftarrow R$ is differentiable almost everywhere, then
\begin{equation}
\label{eq:sum_ig}
    \sum_{j=1}^{|\vec{x}|}IG^{\tau}_j(\vec{x})=F(\vec{x})-F(\vec{b}),
\end{equation}
\end{proposition}
where $j$ is the feature index and $|\vec{x}|$ gives the number of features. $\tau$ represents the straight path between $\vec{x}$ and $\vec{b}$. Deep networks built out of Sigmoids, Relus, and pooling operators satisfy the differentiable condition. Using Equation~\ref{eq:sum_ig} and the definition of $PathIG$ and $IG$ in Equation~\ref{eq:pathig} and~\ref{eq:ig}, we could decompose the $Q_{tot}$ completely to individual contributions through the trajectory path.
\begin{theorem}
\label{prop:add}
Let $\tau_t^{T}$ represents the joint observation and action trajectory from step $t$ to the termination step $T$, then
\begin{equation}
    Q_{tot}(\vec{o}_t,\vec{a}_t)=\sum_{i=1}^{n} \sum_{x_{j}\in \mathbb{X}_{i}}PathIG^{\tau_t^{T}}_{j}(\vec{o},\vec{a}).
\end{equation}
\end{theorem}

\begin{proof}
Let $\vec{x}_t$ represents the feature vector $(\vec{o}_t,\vec{a}_t)$ concisely. $\tau_{t}^{T}$ is composed of $(\tau_{t}^{t+1}, \tau_{t+1}^{t+2}, ..., \tau_{T-1}^{T})$, where $\tau_{t}^{t+1}$ is the straightline path from $(\vec{o}_{t},\vec{a}_{t})$ to $(\vec{o}_{t+1},\vec{a}_{t+1})$.
\begin{equation}\nonumber
\small
\label{eq:valueinvariance}
\begin{split}
     & \;\, Q_{tot}(\vec{o}_t,\vec{a}_t) = Q_{tot}(\vec{x}_t) = Q_{tot}(\vec{x}_t) - Q_{tot}(\vec{x}_{T}) = Q_{tot}(\vec{x}_t) - Q_{tot}(\vec{x}_{t+1}) \\
     & \;\,  + Q_{tot}(\vec{x}_{t+1}) - Q_{tot}(\vec{x}_{t+2}) +...+ Q_{tot}(\vec{x}_{T-1}) - Q_{tot}(\vec{x_{T}})
    \\ & = \sum_{j=1}^{|\vec{x_{t}}|}IG^{\tau_t^{t+1}}_j(\vec{x}) + \sum_{j=1}^{|\vec{x_{t}}|}IG^{\tau_{t+1}^{t+2}}_j(\vec{x}) + 
     ... + \sum_{j=1}^{|\vec{x_{t}}|}IG^{\tau_{T-1}^{T}}_j(\vec{x})
     \\ & = PathIG^{\tau_t^{T}}_{j=1}(\vec{x}) + PathIG^{\tau_t^{T}}_{j=2}(\vec{x}) + ... + PathIG^{\tau_t^{T}}_{j=|\vec{x_{t}}|}(\vec{x}) %=\sum_{j=1}^{|\vec{x_{t}}|} PathIG^{\tau_t^{T}}_{j}(\vec{x})
    % \\ & = \sum_{j}^{|\vec{x_{t}}|} \int_{\alpha=0}^1 \frac{\partial F(\vec{x}_{t+1}+\alpha(\vec{x}_t-\vec{x}_{t+1}))}{\partial (\vec{x}_{t+1}+\alpha(\vec{x}_t-\vec{x}_{t+1}))} \frac{\partial(\vec{x}_{t+1,j}+\alpha(\vec{x}_{t,j}-\vec{x}_{t+1,j}))}{\partial\alpha} \mathrm{d}\alpha + \\ 
    % & \sum_{j}^{|\vec{x}_{t}|} \int_{\alpha=0}^1 \frac{\partial F(\vec{x}_{t+2}+\alpha(\vec{x}_{t+1}-\vec{x}_{t+2}))}{\partial (\vec{x}_{t+2}+\alpha(\vec{x}_{t+1}-\vec{x}_{t+2})))} \frac{\partial(\vec{x}_{t+2,j}+\alpha(\vec{x}_{t+1,j}-\vec{x}_{t+2,j}))}{\partial\alpha} \mathrm{d}\alpha \\
    % & + ... + \sum_{j}^{|\vec{x}_{t}|} \int_{\alpha=0}^1 \frac{\partial F(\vec{x}_{T}+\alpha(\vec{x}_{T-1}-\vec{x}_{T}))}{\partial (\vec{x}_{T}+\alpha(\vec{x}_{T-1}-\vec{x}_{T})))} \frac{\partial(\vec{x}_{T,j}+\alpha(\vec{x}_{T-1,j}-\vec{x}_{T,j}))}{\partial\alpha} \mathrm{d}\alpha
    \\ & = \sum_{x_{j}\in \mathbb{X}_{1}}PathIG^{\tau_t^{T}}_{j}(\vec{x}) + \sum_{x_{j}\in \mathbb{X}_{2}}PathIG^{\tau_t^{T}}_{j}(\vec{x}) + ... + \sum_{x_{j}\in \mathbb{X}_{n}}PathIG^{\tau_t^{T}}_{j}(\vec{x})
    \\ & = \sum_{i=1}^{n} \sum_{x_{j}\in \mathbb{X}_{i}}PathIG^{\tau_t^{T}}_{j}(\vec{x}) = \sum_{i=1}^{n} \sum_{x_{j}\in \mathbb{X}_{i}}PathIG^{\tau_t^{T}}_{j}(\vec{o},\vec{a})
\end{split}
\end{equation}
\end{proof}
Line 4 to line 6 in the proof shows that, as we apply integrated gradients at every adjacent joint state and actions along the trajectory, we aggregate each agent's features' attribution into the contribution of each agent for the global Q-values. Finally, we conclude that integrated gradients on the trajectory path attributes the global Q-value to each agent's feature changes and the decomposition is intact. From the angle of the path integrated gradients, we here find the right feature varying process in RL and then follow the trajectory path to decompose $Q_{tot}$ to individual Q-values on account of each agent's observation and action features.
%Besides the trajectory, the decomposition accuracy also depends on the prediction accuracy of the critic.
%on the condition of the system's joint policies.
%As the global Q-values are attributed to the features of each agent's local observation and action, we could regard the sum of path integrated gradients of agent $i$'s features as agent $i$'s current local state-action value.

\subsection{Multi-channel Critic}
\label{sec:critic}

\begin{figure}[htbp]
\centering
{\includegraphics[height=2.666in,angle=0]{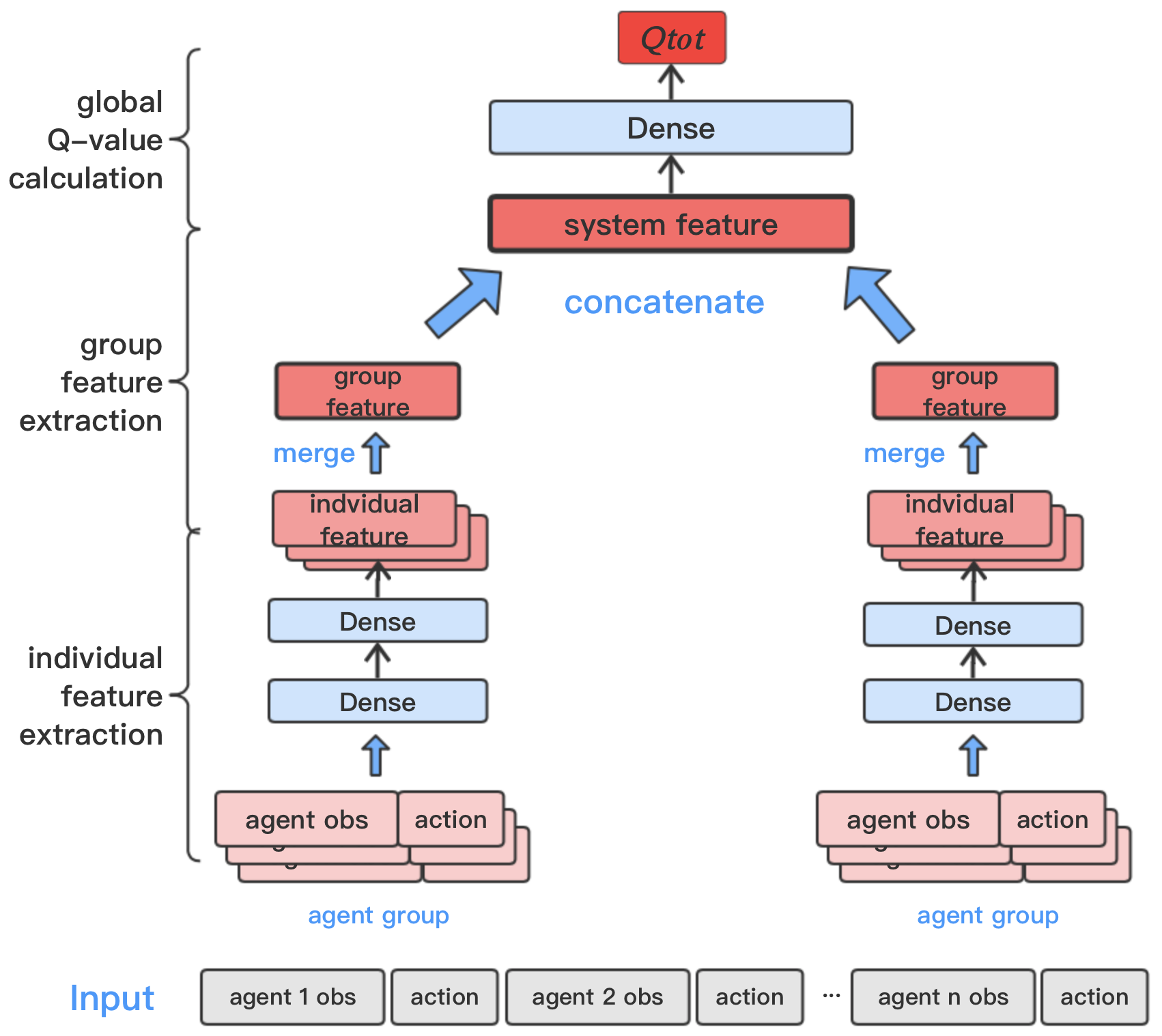}}
\caption{Multi-channel Critic.}
\label{figure:mc-critic}
\end{figure}

In realistic MAS, there may exist heterogeneous agents of different kinds. The space of agents' joint states and actions is very large in such systems, causing the learning of the global Q-value extremely hard. Although agents in MAS are unique, they can also be categorized into different groups according to their feature attributions and personal profile. This fact enlightens us on using sub-network channels to extract information with one channel for one agent group. From bottom to top, agents can be first classified as several kinds of groups and then summarized as a unified system. Based on such a MAS abstraction, we design the multi-channel network structure as illustrated in Figure~\ref{figure:mc-critic} to collect the hidden states from each agent's decentralized observations and actions instead of simply using full-connected layers. At the same time, as there may exist homogeneous agents of the same kind group, we use parameter sharing for homogeneous agents. This technique is adopted widely in many complicated environments and challenging tasks \cite{yang_mean_2018,iqbal_att_maac_2018} and could effectively reduce the network parameters and accelerate learning.

The structure of the critic includes three components: individual feature extraction process, group feature extraction process and the system's global Q-value calculation process. 
%MAS is composed of individual agents of different kind groups. 
We first use the individual feature extracting modules to extract embeddings for agents with one channel responding to one agent group. Next, the group feature merging operation combines the embeddings from the same group and then concatenates them into system features. The merging operation could be either concatenation or addition. Finally, the high-level system features are used to calculate system's Q-values. Following the multi-channel structure, we implicitly represent MAS from decentralized agents to a centralized system. Such a modular critic structure provides a succinct representation of the multiagent Q-value while the number of network parameters can be significantly reduced as well.
%could effectively reduce the number of network parameters without causing information loss. Besides, the modularization design is widely suitable for different kinds of MAS with flexible channel adjustments.

% Another benefit brought by the critic structure is that it helps the correctness of the attribution, for features of one agent, we can divide these features into three parts: $\vec{o_i}, \vec{o_s}, \vec{a_i}$. $\vec{o_i}$ are agent's own features, $\vec{o_s}$ are system's common features, and $\vec{a_i}$ are features of the agent's action.  

\subsection{Algorithm and Training Process}
\label{sec:alg}

\begin{algorithm}[tb]
\caption{Q-value Path Decomposition algorithm}
\label{alg:algorithm}
% \textbf{Input}: Your algorithm's input\\
% \textbf{Parameter}: Optional list of parameters\\
\textbf{Initialize}: Critic network $\theta^{c}$, target critic $\widetilde{\theta}^{c}$ and agents' Q-value networks $\theta^{\boldsymbol{\pi}}=(\theta^{1}, ..., \theta^{n})$
\begin{algorithmic}[1] %[1] enables line numbers
\FOR {each training episode $e$}
\STATE $s_{0}=$ initial state, $t=0$, $h_{0}^{i}=\boldsymbol{0}$ for each agent $i$.
\WHILE{$s_{t} \neq$ terminal \AND $t < T$}
\STATE $t = t + 1$.
\FOR {each agent $i$}
\STATE $Q^{i}(o_{t,i}, \cdot), h_{t}^{i} = $DRQN$(o_{t,i},h_{t-1}^{i};\theta^{i})$.
\STATE Sample $a_{t,i}$ from $\pi_{i}(Q^{i}(o_{t,i}, \cdot),\epsilon(e))$.
\ENDFOR
\STATE execute the joint actions $(a_{t,1}, a_{t,2}, ..., a_{t,n})$.
\STATE receive the reward $r_{t}$ and next state $s_{t+1}$.
\ENDWHILE
\STATE Add episode to buffer and sample a batch of episodes.
\FOR{$e$ in batch}
\FOR{$t=1$ to $T$}
\STATE Calculate targets $y_{t}$ using $\widetilde{\theta}^{c}$.
\ENDFOR
\ENDFOR
\STATE Update critic parameters $\theta^{c}$ with loss $\mathcal{L}(\theta^{c})$.
\STATE Every $C$ episodes reset $\widetilde{\theta}^{c}=\theta^{c}$.
\FOR{$e$ in batch}
\FOR{$t=1$ to $T$}
\STATE Unroll LSTM using states, actions and rewards.
\STATE Using the Integrated Gradients along with the trajectory $e$ to decompose $Q_{tot}$ at time $t$ into $\widetilde{Q}^{i}_{t}=\sum_{x_j\in \mathbb{X}_{i}}PathIG^{\tau}_{x_j}(\vec{o_t},\vec{a_t})$ for each agent $i$.
\ENDFOR
\ENDFOR
\STATE Update $\theta^{\pi}$ with loss $\mathcal{L}(\theta^{i})$ for each agent $i$.
\ENDFOR
\end{algorithmic}
\end{algorithm}

The algorithm details are shown in Algorithm~\ref{alg:algorithm}. Line 2-10 shows that the decentralized agents interact with the environment. Next, Line 13-19 update the critic and target critic networks. The centralized critic $Q_{tot}$ is trained to minimize the loss $\mathcal{L}(\theta^{c})$ as defined in Equation~\ref{eq:critic_loss}.
\begin{equation}
\label{eq:critic_loss}
\begin{split}
\mathcal{L}(\theta^{c}) & = E_{\vec{o}, \vec{a}, r, \vec{o}^{\prime}}[(Q_{tot}^{\theta^{c}}(o_{1},...,o_{n},a_{1},...,a_{n}) - y)^{2}], \\
& y = r + \gamma (Q_{tot}^{\widetilde{\theta}^{c}}(o_{1}^{\prime},...,o_{n}^{\prime},a_{1}^{\prime},...,a_{n}^{\prime}),
\end{split}
\end{equation}
where $\theta^{c}$ is the critic parameters and $\widetilde{\theta}^{c}$ is the target critic parameters, which are reset every $C$ episode. Agent $i$'s network parameters are remarked as $\theta^{i}$. At last, Line 20-26 update each agent's individual Q-value network using the decomposed $\widetilde{Q}^{i}$ as the target label for each agent $i$. The loss of agent $i$'s Q-value network is defined as Equation~\ref{eq:policy_loss}.
\begin{equation}
\label{eq:policy_loss}
\begin{split}
\mathcal{L}(\theta^{i}) & = E_{\vec{o}, \vec{a}, r, \vec{o}^{\prime}}[(Q^{i,\theta^{i}}(o_{i},a_{i}) - \widetilde{Q}^{i})^{2}], \\
& \widetilde{Q}^{i} = \sum_{x_j\in \mathbb{X}_{i}}PathIG^{\tau}_{j}(\vec{o},\vec{a}).
\end{split}
\end{equation}
Notably, for each training, we sample a batch of complete trajectories in the replay buffer for updating. The agent network in the realistic implement is a Recurrent Deep Q-Network (RDQN), which is the basic DQN augmented with the LSTM units. Besides, the exploration policy is $\epsilon$-greedy with $\epsilon(e)$ being the exploration rate as Equation~\ref{eq:epsilon}.
\begin{equation}
\label{eq:epsilon}
\epsilon(e) = \max{(\epsilon_{init} - e * \delta, 0)},
\end{equation}
where $e$ is the episode number. $\epsilon_{init}$ is the start exploration rate and $\delta$ gives the decreasing amount of $\epsilon$ at each episode.

\section{Experiment and Analysis}
\label{section:experiment}

\subsection{Experimental Setup}
In this section, we describe the StarCraft II decentralized micromanagement problems, in which each of the learning agents controls an individual allied army unit. The enemy units are controlled by a built-in StarCraft II AI, which makes use of handcrafted heuristics. The difficulty of the game AI is set to the "very difficult" level. At the beginning of each episode, the enemy units are going to attack the allies. Proper micromanagement of units during battles are needed to maximize the damage to enemy units while minimizing damage received, hence requires a range of skills such as focus fire and avoid overkill. Learning these diverse cooperative behaviors under partial observation is a challenging task, which has become a common-used benchmark for evaluating state-of-the-art MARL approaches such as COMA \cite{foerster_counterfactual_2018}, QMIX \cite{rashid_qmix:_2018} and QTRAN \cite{son_qtran_2019}. We use StarCraft Multi-Agent Challenge (SMAC) environment \cite{samvelyan19smac} as our testbed. More setup details are in the Appendix.

\subsubsection{Network and Training Configurations}
The architecture of agent Q-networks is a DRQN with an LSTM layer with a 64-dimensional hidden state, with a fully-connected layer after, and finally a fully-connected layer with $|A|$ outputs. The input for agent networks is the sequential data which consists of the agent's local observation in recent 12 time steps for all scenarios. The architecture of the QPD critic is a feedforward neural network with the first two dense layers having 64 units for each channel, and then being concatenated or added in each group, and next being concatenated to the output layer of one unit. We set $\gamma$ at 0.99. To speed up learning, we share the parameters across all individual Q-networks and a one-hot encoding of the agent type is concatenated onto each agent’s observations to allow the learning of diverse behaviors. All agent networks are trained using RMSprop with a learning rate of $5\times10^{-4}$ and the critic is trained with Adam with the same learning rate. Replay buffer contains the most recent 1000 trajectories and the batch size is 32. Target networks for the global critic are updated after every 200 training episodes.

\subsubsection{Decomposition Path Settings}
For the Q-value decomposition process, integrated gradients can be efficiently approximated via a summation at points occurring at sufficiently small intervals along the trajectory path over each pair of consecutive state-action transitions $(\vec{o}_{t},\vec{a}_{t})$ and $(\vec{o}_{t+1},\vec{a}_{t+1})$. Then the gradient integral path is obtained by repeatedly interpolating between every two adjacent states from the current state to the terminated state. With $m$ being the number of steps in the Riemman approximation and $\vec{x}_t$ being $(\vec{o}_{t},\vec{a}_{t})$ for simplification, we calculate the integrated gradients for every two adjacent states as:
\begin{equation}
\small
\begin{split}
    & \widetilde{IG}^{\tau_t^{t+1}}_{j}(\vec{o}_{t},\vec{a}_{t}) = \widetilde{IG}^{\tau_t^{t+1}}_{j}(\vec{x}_t) ::= \\
    & (\vec{x}_{t,j} - \vec{x}_{t+1,j}) \times \sum_{k=1}^m \frac{\partial F(\vec{x}_{t+1}+\frac{k}{m} \times (\vec{x}_t-\vec{x}_{t+1}))}{\partial (\vec{x}_{t+1}+\frac{k}{m} \times (\vec{x}_t-\vec{x}_{t+1}))} \times \frac{1}{m}.
\end{split}
\end{equation}
Although larger $m$ could obtain more accurate decomposition, due to the trade-off of high qualified performance and limited computation time and resources, we set $m$ at 5 after experimental study but it has exhibited impressive performance, which could be referred in Section~\ref{sec:decom_step}.

\subsection{Results}
To validate QPD, we evaluate it on both homogeneous and heterogeneous scenarios. To encourage exploration, we use $\epsilon$-greedy which anneals from 1 to 0 at the first 2000 episodes. We test our method at every 100 training episodes on 100 testing episodes with exploratory behaviors disabled. The main evaluation metric is the win percentage of evaluation episodes over the course of training \cite{samvelyan19smac}.
%This metric can be estimated by periodically running a fixed number of evaluation episodes with exploratory behaviors disabled.
The results include the median performance as well as the 25-75\% percentiles recommended in \cite{samvelyan19smac} to avoid the effect of any outliers. Another metric, the mean win rate over all runs, is also reported. All experiments are conducted across 12 independent runs and QPD's learning curves on all maps are shown in Figure~\ref{figure:results}.
\begin{figure}[htbp]
\centering
%\captionsetup{font={small}}
\subfigure[\scriptsize{Map 3m}]{
\label{3m}
\includegraphics[width=0.18\textwidth]{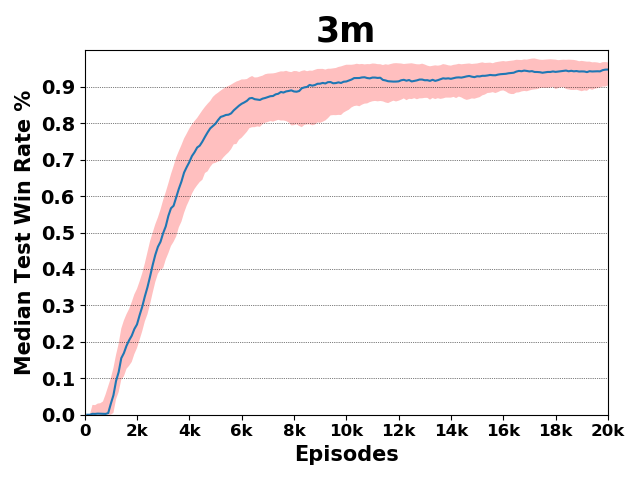}}
\subfigure[\scriptsize{Map 8m}]{
\label{8m}
\includegraphics[width=0.18\textwidth]{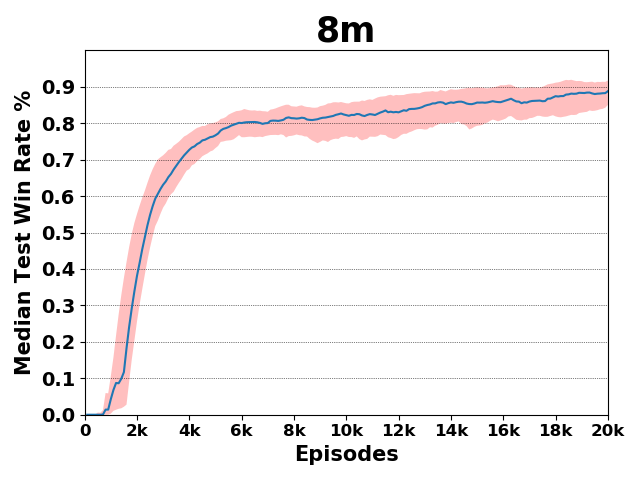}}
\subfigure[\scriptsize{Map 2s3z}]{
\label{2s3z}
\includegraphics[width=0.18\textwidth]{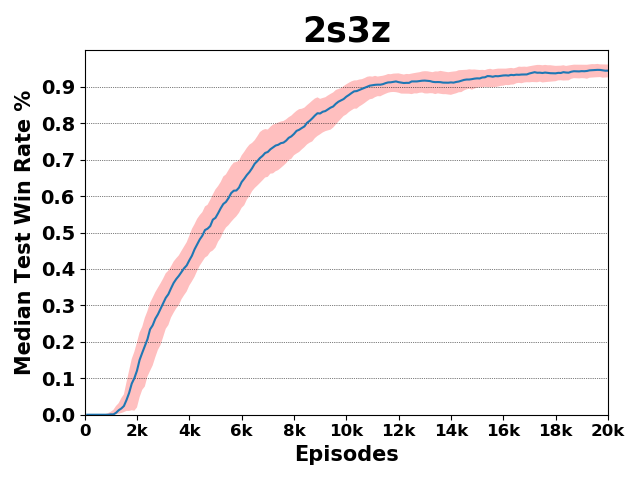}}
\subfigure[\scriptsize{Map 3s5z}]{
\label{3s5z}
\includegraphics[width=0.18\textwidth]{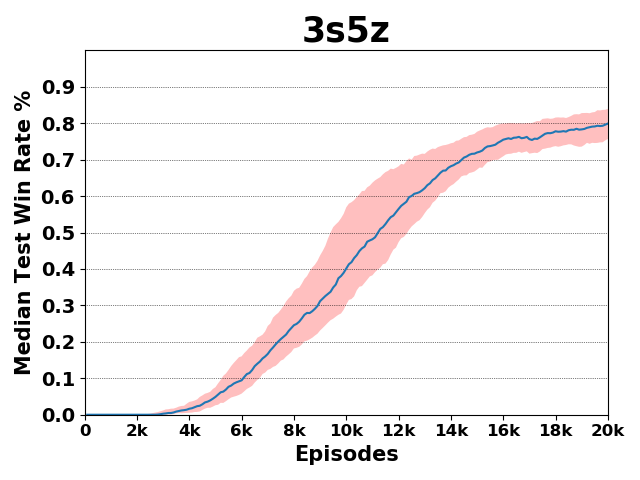}}
% \subfigure[\scriptsize{Map MMM2}]{
% \label{MMM2}
% \includegraphics[width=0.23\textwidth]{MMM2.png}}
\subfigure[\scriptsize{Map 1c3s5z}]{
\label{1c3s5z}
\includegraphics[width=0.18\textwidth]{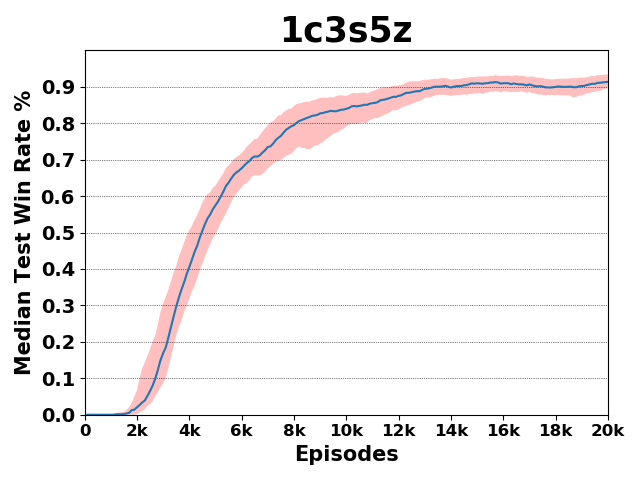}}
\subfigure[\scriptsize{Map 3s5z\_vs\_3s6z}]{
\label{3s5z_vs_3s6z}
\includegraphics[width=0.18\textwidth]{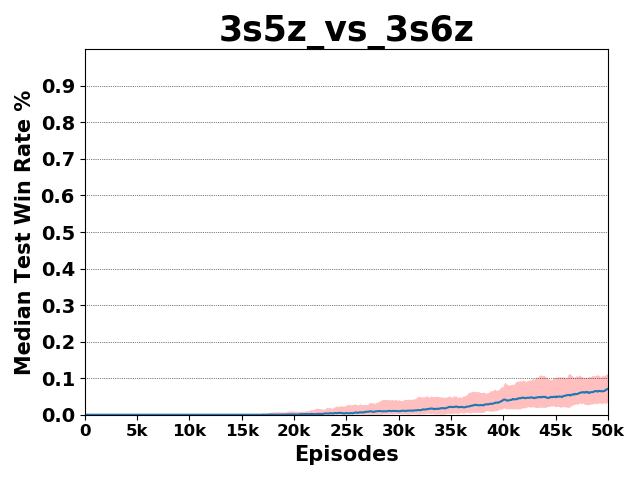}}
\caption{QPD's median win percentage of different map scenarios. 25\%-75\% percentile is shaded.}
\label{figure:results}
\end{figure}

% \begin{figure}[H]
% \centering
% %\captionsetup{font={small}}
% \subfigure[\scriptsize{Map 2s3z}]{
% \label{2s3z}
% \includegraphics[width=0.232\textwidth]{2s3z.png}}
% \subfigure[\scriptsize{Map 3s5z}]{
% \label{3s5z}
% \includegraphics[width=0.232\textwidth]{3s5z.png}}
% \caption{Median win percentage of heterogeneous map scenarios. 25\%-75\% percentile is shaded.}
% \label{figure:heterresults}
% \end{figure}

All maps are of the different agent number or different types. Both sides in Map 3m have 3 Marines while in Map 8m have 8 Marines. In Map 2s3z, both sides have 2 Stalkers and 3 Zealots. For Map 3s5z, both sides have 3 Stalkers and 5 Zealots. Map 1c3s5z, both sides have an extra Colossus compared with Map 3s5z.
In map 3s5z\_vs\_3s6z, ally has 3 Stalkers and 5 Zealots while enemy has 3 Stalkers and 6 Zealots. 
%In Map MMM2, ally has 1 Medivac, 2 Marauders and 7 Marines while enemy has 1 Medivac, 3 Marauders and 8 Marines.
To compare QPD with existing MARL methods, we use results from SMAC \cite{samvelyan19smac} because methods in their report show higher performance than the original works \cite{rashid_qmix:_2018,foerster_counterfactual_2018} and our implementation. We also compare with QTRAN. Table~\ref{tab:results} shows the evaluation metric results, where $\widetilde{m}$ is the median win percentage and $\overline{m}$ is the mean win percentage.
% of all approaches. QMIX, COMA, and IDL are all reported from SMAC.

\begin{table}[ht]
\scriptsize
\centering
\caption{Median and mean performance of the test win percentage.}
\label{tab:results}
\begin{tabular}{|c|c|c|c|c|c|c|c|c|c|c|}
\hline
\multirow{2}{*}{Map}                                         & \multicolumn{2}{c|}{IQL} & \multicolumn{2}{c|}{COMA} & \multicolumn{2}{c|}{QMIX} & \multicolumn{2}{c|}{QTRAN} & \multicolumn{2}{c|}{QPD} \\ \cline{2-11} 
                                                             & $\small{\widetilde{m}}$       & $\small{\overline{m}}$      & $\small{\widetilde{m}}$       & $\small{\overline{m}}$       & $\small{\widetilde{m}}$       & $\small{\overline{m}}$       & $\small{\widetilde{m}}$        & $\small{\overline{m}}$       & $\small{\widetilde{m}}$       & $\small{\overline{m}}$      \\ \hline
3m                                                           & \textbf{100} & 97 & 91 & 92 & \textbf{100} & 99 & \textbf{100} & \textbf{100} & 95 & 92 \\ \hline
8m                                                           & 91 & 90 & 95 & 94 & \textbf{100} & 96 & \textbf{100} & \textbf{97} & 94 & 93 \\ \hline
2s3z                                                         & 39 & 42 & 66 & 64 & \textbf{100} & \textbf{97} & 77 & 80 & 95 & 94 \\ \hline
3s5z                                                         & 0 & 3 & 0 & 0 & 16 & 25 & 0 & 4 & \textbf{85} & \textbf{81} \\ \hline
1c3s5z                                                       & 7 & 8 & 30 & 30 & 89 & 89 & 31 & 33 & \textbf{92} & \textbf{92} \\ \hline
\begin{tabular}[c]{@{}c@{}}3s5z\\ \_vs\_\\ 3s6z\end{tabular} & 0 & 0 & 0 & 0 & 0 & 0 & 0 & 0 & \textbf{8} & \textbf{10} \\ \hline
\end{tabular}
\end{table}

We could see that QPD's performance is competitive with QMIX in three simple scenarios, 3m, 8m, 2s3z and 1c3s5z. More importantly, in the more difficult 3s5z where all existing methods perform poorly, QPD achieves superior performance much better than others. Furthermore, in a super hard scenarios 3s5z\_vs\_3s6z, QPD also beats other methods, where all other methods fail completely. To understand the rationale behind the results, we analyze the learned behaviors of agents. In 3m, agents learn the micro focus fire for beating enemies. Furthermore, in 8m, agents learn to stand into a line to shoot the enemy while avoiding overkill. In the heterogeneous 2s3z and 1c3s5z, both QMIX and QPD could solve it. Our method successfully learned to intercept the enemy Zealots with allied Zealots to protect the allied Stalkers from severe damage. However, in 3s5z, the learned policy of QPD is quite different from 2s3s: allied Zealots go around the enemy Zealots to attack the enemy Stalkers first and then attack the enemy Zealots with the allied Stalkers on both sides. Such a highly coordinated policy cannot be learned by QMIX \cite{samvelyan19smac}. 
%In MMM2, the key learned strategy to winning is that the Medivac is heading to enemies first to absorbing fire and heal other allies at the same time.
In 3s5z\_vs\_3s6z, Zealots need to hold enemy's Zealots to protect ally's Stalkers and attack enemy's Stalkers at the same time. Such a behaviour is learned only by QPD which starts to win. Overall, QPD learns excellent decentralized policies comparable to the state-of-the-art MARL methods in both homogeneous and heterogeneous scenarios and outperforms QMIX and QTRAN in more complicated settings.

\subsection{Ablation}
\subsubsection{Multi-channel Critic Evaluation}
Using a modular network structure in the centralized critic is common in MARL algorithms and could effectively improve the performance \cite{iqbal_att_maac_2018,liu_pic_2019}. We also test the naive critic with several fully-connected dense layers, but we found this structure is with a high variance and its performance is lower than the modular ones. The reason is that the number of features fed into the critic is up to hundreds and increases quadratically with the number of agents, which causes a huge challenge for the naive network to learn effective hidden states from these features. Thus, we omit the naive critic's results. One main difference with previous modular critic methods is that, we explicitly consider the heterogeneous multiagent setting. We use different channels for different kinds of agents. Furthermore, we choose the concatenation operation as the way of the hidden features integration from each channel. We show this design could slightly improve the performance of QPD. The reason behind this phenomenon is clear. The multi-channel and concatenation operation own the greater representation ability to keep track of the feature influence of each agent of each kind in the multiagent Q-value prediction process.

\begin{figure}[htbp]
\centering
%\captionsetup{font={small}}
\subfigure[\scriptsize{Map 3m}]{
\label{3m-op-ab}
\includegraphics[width=0.23\textwidth]{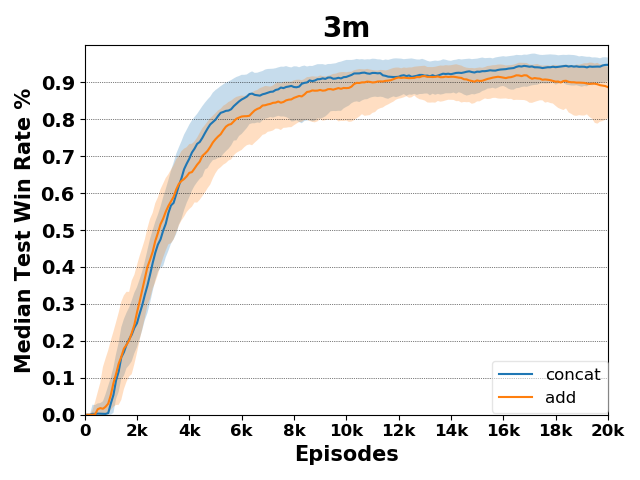}}
\subfigure[\scriptsize{Map 2s3z}]{
\label{2s3z-op-ab}
\includegraphics[width=0.23\textwidth]{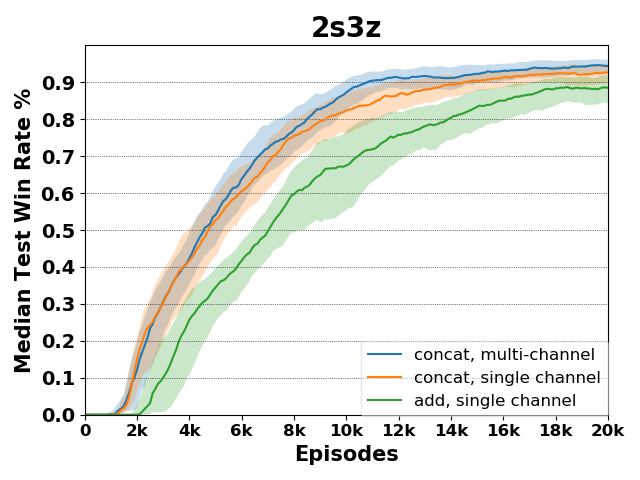}}
\caption{Median win percentage of 12 runs for critic ablation.}
\label{figure:critic_ablation}
\end{figure}

\begin{figure}[htbp]
\centering
%\captionsetup{font={small}}
\subfigure[\scriptsize{Map 3m}]{
\label{3m-ds-ab}
\includegraphics[width=0.23\textwidth]{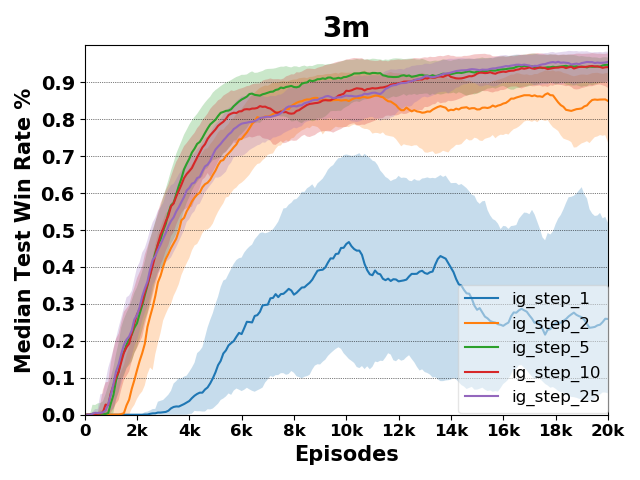}}
\subfigure[\scriptsize{Map 2s3z}]{
\label{2s3z-ds-ab}
\includegraphics[width=0.23\textwidth]{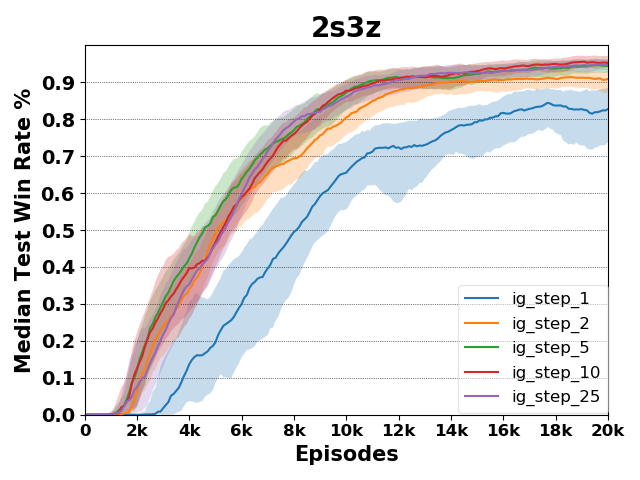}}
\caption{Median win percentage of 12 runs for decomposing steps.}
\label{figure:decom_step_ablation}
\end{figure}

\subsubsection{Decomposition Step}
\label{sec:decom_step}
As the integrated gradients is the core of QPD, it is critical and interesting to study the decomposition step's impact on the performance. Between each adjacent joint state-action pairs, we set the decomposition step of 1, 2, 5, 10 and 25 for studying. Results are presented in Figure~\ref{figure:decom_step_ablation}. As we can see, the decomposition step affects the performance a lot. When the decomposition steps is low, the decomposition is not accurate enough to assign credits for agents, thus making the training unstable and win rate low. But when the decomposition step increases, the more accurate decomposed individual Q-values could update the policies more accurately. Especially, QPD is capable of the setting of moderate decomposition step number, where step of 5 could reach a comparable performance level of step 10 and 25. It means that QPD does not require lots of computation resources for decomposing to reach a high performance.

%Table~\ref{tab:config} shows the training configurations of QMIX \cite{rashid_qmix:_2018}, COMA \cite{foerster_counterfactual_2018} and IDL in SMAC compared with QPD. We can see that QPD achieves similar performance of QMIX with significantly less training steps and less training resource requirements. The superior training efficiency of QPD can be explained from the following two aspects. First, QPD makes full use of the trajectories effectively to compute the gradient update signals for each agent. The second reason is that the design of multi-channel critic which significantly reduces the number of network parameters accelerates the training speed.
 
% \begin{table}[ht]
% \centering
% \small
%   \caption{Training configurations.}
%   \label{tab:config}
%   \begin{tabular}{ccc}
%     \toprule
%     \multicolumn{1}{c}{Training Schema} & \multicolumn{1}{c}{\begin{tabular}[c]{@{}c@{}} SMAC \end{tabular}} & \multicolumn{1}{c}{\begin{tabular}[c]{@{}c@{}} Ours \end{tabular}} \\
%     \midrule
%     \textbf{Training Length} & \textbf{10 million steps} & \multicolumn{1}{c}{\begin{tabular}[c]{@{}c@{}} \textbf{20000 episodes} \\ (\textbf{$<$ 3 million steps}) \end{tabular}} \\
%     Replay Buffer & 5000 & 1000 \\
%     Sequential Length & full & 12 \\
%     Batch Size & 32 & 32 \\
%   \bottomrule
% \end{tabular}
% \end{table}

\section{Conclusion and Future Work}
\label{section:conclusion}
In this paper, we propose QPD to solve the multiagent credit assignment problem in Dec-POMDP settings. Different from previous methods, we propose the trajectory-based integrated gradients attribution method to achieve effective Q-value decomposition at the agent level. Experiments on the challenging StarCraft II micromanagement tasks show that QPD learns well coordinated policies on various scenarios and reaches the state-of-the-art performance.

For the future work, better configurations of the path integrated gradients should be investigated to help attribution such as alternative choices of interpolation methods. Also, policy gradient methods combined with the path integrated gradients is expected to leverage better coordination.

% \section{Related Work}
% \label{section:related_work}
%We evaluate our algorithm with the StarCraft\uppercase\expandafter{\romannumeral2} as the use case to explore the learning of intelligent collaborative strategies among multiple agents, which has recently emerged as a challenging RL benchmark task with high stochasticity, a large state-action space, and delayed rewards. In more detail, we focus on StarCraft\uppercase\expandafter{\romannumeral2} micromanagement tasks, where each player controls their own units to destroy the opponents army in the combats under different terrain conditions. To demonstrate the effectiveness of QPD, we follow the challenging game settings of \cite{rashid_qmix:_2018}, which massively reduces each agent’s field-of-view and removes access to these macro-actions to make the game more realistic and intricate. Our experiments show that QPD helps multi-agent systems learn complicated cooperation policies in both heterogeneous and asymmetric scenarios.
%Additionally, by treating Dec-POMDPs as Occupancy-MDPs, oSARSA is proposed in \cite{dibangoye_learning_2018} and a policy-gradient approach is proposed in \cite{bono_cooperative_2019}. As they focus on different MARL challenges and are applicable for different task scenarios, we do not discuss them in details.

\bibliography{example_paper}
\bibliographystyle{icml2020}

\clearpage
\appendix{}
\appendixpage

\section{Environment Settings}
\subsection{States and Observations}
We mainly follow the settings of SMAC \cite{samvelyan19smac}. At each time step, agents receive local observations within their field of view. This encompasses information about the map within a circular area around each unit with a radius equal to the sight range. The sight range makes the environment partially observable for each agent. An agent can only observe other agents if they are both alive and located within its sight range. Hence, there is no way for agents to distinguish whether their teammates are far away or dead. The feature vector observed by each agent contains the following attributes for both allied and enemy units within the sight range: distance, relative x, relative y, health, shield, and unit type. All Protos units have shields, which serve as a source of protection to offset damage and can regenerate if no new damage is received. The global state is composed of the joint observations but removing the restriction of sight range, which could be obtained during training in the simulations. All features, both in the global state and in individual observations of agents, are normalized by their maximum values.

\subsection{Action Space}
We follow the settings of SMAC \cite{samvelyan19smac}. The discrete set of actions which agents are allowed to take consists of move[direction], attack[enemy id], stop and no-op. Dead agents can only take no-op action while live agents cannot. Agents can only move with a fixed movement amount 2 in four directions: north, south, east, or west. To ensure decentralization of the task, agents are restricted to use the attack[enemy id] action only towards enemies in their shooting range. This additionally constrains the ability of the units to use the built-in attack-move macro-actions on the enemies that are far away. The shooting range is set to be 6 for all agents. Having a larger sight range than a shooting range allows agents to make use of the move commands before starting to fire. The unit behavior of automatically responding to enemy fire without being explicitly ordered is also disabled.

\subsection{Rewards}
We follow the settings of SMAC \cite{samvelyan19smac}. At each time step, the agents receive a joint reward equal to the total damage dealt on the enemy units. In addition, agents receive a bonus of 10 points after killing each opponent, and 200 points after killing all opponents for winning the battle. The rewards are scaled so that the maximum cumulative reward achievable in each scenario is around 20.

\section{Hyper-parameters}
The hyper-parameters of QPD are shown in Table~\ref{tab:hyper-para}, including training configurations and network configurations. More details could be referred in the provided source codes. Specially, the total training episode number for the 3s5z\_vs\_3s6z is 50000 while all other maps' total training episode number is 20000 as shown in the table. The architectures of agents' RDQN network and QPD's critic network are the same as shown in Figure~\ref{figure:framework} and Figure~\ref{figure:mc-critic} respectively.

\begin{table}[htbp]
\caption{Hyper-parameters of QPD.}
\label{tab:hyper-para}
\begin{tabular}{c|c|c}
\hline
Setting                                                                             & Name                   & Value         \\ \hline
\multirow{13}{*}{\begin{tabular}[c]{@{}c@{}}Training\\ configurations\end{tabular}} & Replay buffer size     & 1000 episodes \\
                                                                                    & Batch size             & 32 episodes   \\
                                                                                    & Total training episodes  & 20000         \\
                                                                                    & Exploration episodes   & 2000          \\
                                                                                    & Start exploration rate & 1.0           \\
                                                                                    & End exploration rate   & 0.0           \\
                                                                                    & Agent input length     & 12 steps      \\
                                                                                    & Gamma                  & 0.99          \\
                                                                                    & Target update interval & 200 episodes  \\
                                                                                    & Parallel environment   & 8             \\
                                                                                    & Training interval      & 100 episodes  \\
                                                                                    & Testing battle number  & 100 episodes  \\
                                                                                    & Decomposition step     & 5             \\ \hline
\multirow{5}{*}{\begin{tabular}[c]{@{}c@{}}Network\\ configurations\end{tabular}}   & Agent learning rate    & 0.0005        \\
                                                                                    & Critic learning rate   & 0.0005        \\
                                                                                    & Agent RDQN optimizer   & RMSProp        \\
                                                                                    & Critic optimizer       & Adam        \\
                                                                                    & Channel dense unit     & 64            \\
                                                                                    & LSTM unit              & 64            \\
                                                                                    & Clipping global norm   & 5             \\ \hline
\end{tabular}
\end{table}

\end{document}